\documentclass[a4paper,UKenglish]{lipics-v2016}
 
\usepackage{microtype}

\usepackage{pgf, tikz}
\usetikzlibrary{arrows, automata, petri}

\usepackage{longtable}

\usepackage{algorithm}
\usepackage[noend]{algpseudocode}

\title{The \io complexity of Strassen's matrix multiplication
  with recomputation}
\titlerunning{The \io complexity of Strassen's matrix multiplication
  with recomputation} 

\author[1]{Gianfranco Bilardi}
\author[2]{Lorenzo De Stefani}
\affil[1]{Department of Information Engineering, University of Padova, Via Gradenigo 6B/Padova, Italy\\
  \texttt{bilardi@dei.unipd.it}}
\affil[2]{Department of Computer Science, Brown University, 115 Waterman Street/Providence, United States of America\\
  \texttt{lorenzo@cs.brown.edu}}
\authorrunning{G. Bilardi and L. De Stefani} 

\Copyright{Gianfranco Bilardi and Lorenzo De Stefani}


\newcommand{\ri}{\mathcal{R}}

\newcommand{\BOme}[1]{\Omega\left(#1\right)}

\newcommand{\io }{I/O }
\newcommand{\stral}[1]{\mathcal{H}^{#1}}
\newcommand{\alg}{\mathcal{A}}

\begin{document}

\maketitle

\begin{abstract}
A tight $\Omega((n/\sqrt{M})^{\log_2 7}M)$ lower bound is derived on
the \io complexity of Strassen's algorithm to multiply two $n \times
n$ matrices, in a two-level storage hierarchy with $M$ words of fast
memory.  A proof technique is introduced, which exploits the Grigoriev's
flow of the matrix multiplication function as well as some
combinatorial properties of the Strassen computational directed acyclic graph (CDAG).
Applications to parallel computation are also developed. The result
generalizes a similar bound previously obtained under the constraint
of no-recomputation, that is, that intermediate results cannot be
computed more than once. For this restricted case, another lower bound
technique is presented, which leads to a simpler analysis of the \io
complexity of Strassen's algorithm and can be readily extended to
other ``Strassen-like'' algorithms.
 \end{abstract}

\section{Introduction}
Data movement is increasingly playing a major role in the performance
of computing systems, in terms of both time and energy. This current
technological trend~\cite{patterson2005getting} is destined to
continue, since the very fundamental physical limitations on minimum
device size and maximum message speed lead to inherent costs when
moving data, whether across the levels of a hierarchical memory system
or between processing elements of a parallel
systems~\cite{bilardi1995horizons}. The communication requirements of
algorithms have been the target of considerable research in the last
four decades; however, obtaining significant lower bounds based on such
requirements remains and important and challenging task.

In this paper we focus on the \io complexity of Strassen's matrix
multiplication algorithm.  Matrix multiplication is a pervasive
primitive utilized in many applications.
Strassen~\cite{strassen1969gaussian} showed that two $n \times n$
matrices can be multiplied with $O(n^{\omega})$, with $\omega = \log_2
7 \approx 2.8074$, hence with asymptotically fewer than the $n^3$
arithmetic operations required by the straightforward implementation
of the definition of matrix multiplication. This result has motivated
a number of efforts which have lead to increasingly faster algorithms,
at least asymptotically, with the current record being at $\omega <
2.3728639$~\cite{legall2014}.

\paragraph*{Previous and related work}

\io complexity has been introduced in the seminal work by Hong and
Kung~\cite{jia1981complexity}; it is essentially the number of data
transfers between the two levels of a memory hierarchy with a fast
memory of $M$ words and a slow memory with an unbounded number of
words. They presented techniques to develop lower bounds to the \io
complexity of computations modeled by \emph{computational directed acyclic graphs} (CDAGs). The resulting lower bounds
apply to all the schedules of the given CDAG, including those with
recomputation, that is, where some vertices of the CDAG are evaluated
multiple times. Among other results, they established an $\BOme{n^3/
  \sqrt{M}}$ lower bound to the \io complexity of the definition-based
matrix multiplication algorithm, which matched a known upper bound
upper bound~\cite{cannon1969cellular}. The techniques
of~\cite{jia1981complexity} have also been extended to obtain tight
communication bounds for the definition-based matrix multiplication in
some parallel settings~\cite{irony2004communication,frigo1999cache}.

Ballard et al.  generalized the results on matrix multiplication of
Hong and Kung~\cite{jia1981complexity} in~\cite{ballard2011minimizing,
  ballard2010communication} by using the approach proposed
in~\cite{irony2004communication} based on the Loomis-Whitney geometric
theorem~\cite{loomis1949,zalg}. The same papers presents tight \io
complexity bounds for various classical linear algebra algorithms, for
problems such as LU/Cholesky/LDLT/QR factorization and eigenvalues and
singular values computation. 
 
It is natural to wonder what is the impact of Strassen's reduction of
the number of arithmetic operations on the number of data transfers.
In an important contribution Ballard et
al.~\cite{ballard2012graph}, obtained an
$\Omega((n/\sqrt{M})^{\log_2 7}M)$ \io lower bound for Strassen's
algorithm, using the ``\emph{edge expansion approach}''. The authors
extend their technique to a class of ``\emph{Strassen-like}'' fast
multiplication algorithms and to fast recursive multiplication
algorithms for rectangular matrices~\cite{ballard2012graphrec}. This
result was later generalized to a broader class of
``\emph{Strassen-like}'' algorithms of by Scott
et. al~\cite{scott2015matrix} using the ``\emph{path routing}''
technique.  A parallel, ``\emph{communication avoiding}''
implementation of Strassen's algorithm whose performance matches the
known lower bound~\cite{ballard2012graph,scott2015matrix}, was
proposed by Ballard et al.~\cite{ballard2012communicationalg}.

The edge expansion technique of~\cite{ballard2012graph}, the path
routing technique of~\cite{scott2015matrix}, and the ``\emph{closed dichotomy
width}'' technique of~\cite{bilardi1999processor} all yield \io lower
bounds that apply only to computational schedules for which no
intermediate result is ever computed more than once
(\emph{nr-computations}). 
While it is of interest to know what is the
\io complexity achievable by nr-computations, it is also
important to investigate what can be achieved with recomputation. In
fact, for some CDAGs, recomputing intermediate values does reduce the
space and/or the \io complexity of an
algorithm~\cite{savage1995extending}.
In~\cite{bilardi2001characterization}, it is shown that some
algorithms admit a \emph{portable schedule} (i.e., a schedule which
achieves optimal performance across memory hierarchies with different
access costs) only if recomputation is allowed. A number of lower
bound techniques that allow for recomputation have been presented in
the literature, including the ``\emph{S-partition}
technique''~\cite{jia1981complexity}, the ``\emph{S-span}
technique''~\cite{savage1995extending}, and the ``\emph{S-covering}
technique''~\cite{bilardi2000space} which merges and extends aspects
from both~\cite{jia1981complexity}
and~\cite{savage1995extending}. None of these has however been
successfully applied to fast matrix multiplication algorithms.

\paragraph*{Our results}

Our main result is the extension of the $\Omega((n/\sqrt{M})^{\log_2
  7}M)$ \io complexity lower bound for Strassen's algorithm to
schedules with recomputation. A matching upper bound is known, and
obtained without recomputation; hence, we can conclude that, for
Strassen's algorithm, recomputation does not help in reducing \io
complexity if not, possibly, by a constant factor. In addition to the
result itself, the proof technique appears to be of independent
interest, since it exploits to a significant extent the
divide\&conquer nature which is exhibited by many algorithms.  We do
follow the dominator-set approach pioneered by Hong and Kung
in~\cite{jia1981complexity}. However, we focus the dominator analysis
only on a select set of target vertices, specifically the outputs of
the sub-CDAGs of Strassen's CDAG that correspond to sub-problems of a
suitable size (a size chosen as a function of the fast memory
capacity, $M$).  Any dominator set of a set of target vertices can be
partitioned into two subsets, one internal and one external to the
sub-CDAGs. The analysis of the external component of the dominator does
require rather elaborate arguments that are specific to Strassen's
CDAG. In contrast, the analysis of the internal component can be
carried out based only on the fact that the sub-CDAGs compute matrix
products, irrespective of the algorithm (in our case, Strassen's) by
which the products are computed.  To achieve this independence of the
algorithm, we resort on the concept of Grigoriev's flow of a
function~\cite{grigor1976application} and on a lower bound to such
flow established by Savage~\cite{savage97models} for matrix
multiplication.

As it turns out, for schedules without recomputation, the analysis of
the external component of the dominator sets becomes unnecessary and
the analysis of the internal components can be somewhat simplified.
The result is a derivation of the \io complexity without recomputation
considerably simpler than those in~\cite{ballard2012graph,
  scott2015matrix}.  The technique is easily extended to a class of
Strassen-like algorithms.

The paper is organized as follows: in Section~\ref{sec:preliminaries},
we provide the details of our model and of several theoretical notions
needed in our analysis. Section~\ref{sect:stranore} describes the
simplified lower bound for Strassen's and Strassen-like algorithms,
when implemented with no recomputation. In Section~\ref{sec:stragen}, we present
the \io complexity lower bound for Strassen's algorithm, when
recomputation is allowed. Extensions of the result to a parallel
model are also discussed.

\section{Preliminaries}\label{sec:preliminaries}
We consider algorithms which compute the product $C=AB$ of $n \times
n$ matrices $A,B$ with entries form a ring $\ri$. Specifically, we
focus on algorithms whose execution, for any given $n$, can be modeled
as a \emph{computational directed acyclic graph} (CDAG) $G(V,E)$,
where each vertex $v\in V$ represents either an input value or the
result of an unit time operation (i.e.  an intermediate result or one
of the output values), while the directed edges in $E$ represent data
dependences. A \emph{directed path} connecting vertices $u,v\in V$ is
an ordered sequence of vertices for which $u$ and $v$ are respectively
the first, and last vertex such that there is in $E$ a (directed) edge
pointing from each vertex in the sequence to its successor. We say
that a CDAG $G'(V',E')$ is a \emph{sub-CDAG} of $G(V,E)$ if
$V'\subseteq V$ and $E' \subseteq E$.
\paragraph*{Model}

We assume that sequential computations are executed on a system with a
two-level memory hierarchy consisting of a fast memory or \emph{cache}
of size $M$, measured in words, and a \emph{slow memory} of
unlimited size. We assume that each memory word can store at most one
value form $\ri$. An operation can be executed only if all its
operands are presents in cache. Data can be moved from the slow memory
to the cache by \emph{read} operations and in the other direction
by \emph{write} operations. Read and write operations are also called
\emph{\io operations}. We assume the input data to be stored in slow
memory at the beginning of the computation. The evaluation of a CDAG in this model can be analyzed by means of the ``\emph{red-blue
  pebble game}''~\cite{jia1981complexity}. The number of \io
operations executed when evaluating a CDAG depends on the
``\emph{computational schedule},'' that is, on the order in which
vertices are evaluated and on which values are kept in/discarded from
cache.  The \emph{\io complexity} $IO_{G}(M)$ of CDAG $G$ is defined
as the minimum number of \io operations over all possible
computational schedules.

We also consider a parallel model where $P$ processors, each with
a local memory of size $M$, are connected by a network. We assume that
the input is initially distributed among the processors, thus
requiring that $MP\geq 2n^2$. Processors can exchange point-to-point
messages among each other. For this model, we derive lower bounds to
the number of words that must be either sent or received by at least
one processor during the CDAG evaluation.

\paragraph*{Preliminary definitions}
The concept of \emph{information flow of a function} was originally
introduced by Grigoriev~\cite{grigor1976application}. We use a revised
formulation presented by Savage~\cite{savage97models}. We remark that
the flow is an inherent property of a function, not of a specific
algorithm by which the function may be computed.
\begin{definition}[Grigoriev's flow of a function]
A function $f:\mathcal{\ri}^p\rightarrow\mathcal{\ri}^q$ has a $w\left(u, v\right)$ Grigoriev's flow if for all subsets $X_1$ and $Y_1$ of its $p$ input and $q$ output variables, with $|X_1| \geq u$ and $|Y_1| \geq  v$, there is a sub-function $h$ of $f$ obtained by making some assignment to variables of $f$ not in $X_1$  and discarding output variables not in $Y_1$ such that $h$ has at least $|\ri|^{w(u,v)}$ points in the image of its domain. 	
\end{definition}
A lower bound on the Grigoriev's flow for the square matrix multiplication function $f_{n \times n}: \ri^{2n^2} \rightarrow \ri^{n^2}$ over the ring $\ri$ was presented by Savage in~\cite{savage97models} (Theorem 10.5.1). 
 
\begin{lemma}[Grigoriev's flow of $f_{n \times n}: \ri^{2n^2} \rightarrow \ri^{n^2}$ ~\cite{savage97models}]\label{lem:info_flo_mat_mul}
	$f_{n \times n}: \ri^{2n^2} \rightarrow \ri^{n^2}$
 has a $w_{n\times n}\left(u, v\right)$ Grigoriev's flow, where:
	\begin{equation}
		w_{n \times n}\left(u, v\right) \geq \frac{1}{2}\left(v-\frac{\left(2n^2 -u\right)^2}{4n^2}\right),\text{ for }0\leq u \leq 2n^2,\ 0\leq v\leq n^2.
	\end{equation}
\end{lemma}

The \emph{dominator set} concept was originally introduced in~\cite{jia1981complexity}. 

\begin{definition}[Dominator set]
Given a CDAG $G(V,E)$, let $I\subset V$ denote the set of input
vertices. A \emph{dominator set} for $V'\subseteq V$ is a set $\Gamma
\subseteq V$ such that every path from a vertex in $I$ to a vertex in
$V'$ contains at least a vertex of $\Gamma$. A \emph{minimum dominator
  set} for $V'$ is a dominator set with minimum cardinality.
\end{definition}

We will use a specular concept, commonly referred as
``\emph{post-dominator set}''.

\begin{definition}[Post-dominator set]
Given a CDAG $G(V,E)$, let $O\subset V$ denote the set of output vertices. A \emph{post-dominator set} for $V'\subseteq V\setminus O$ with respect to $O'\subseteq O$ is defined to be a set of vertices in $V$ such that every path from a vertex in $V'$ to the output vertices in $O'$ contains at least a vertex of the set.
A \emph{minimum post-dominator} set for $V'\subseteq V\setminus O$ with respect to $O'\subseteq O$ is a post-dominator set with minimum cardinality.
\end{definition}

\noindent The following lemma relates the size of certain post-dominator sets to
the Grigoriev's flow.

\begin{lemma}\label{lem:flowpost}
Let $G(V,E)$ be a CDAG computing $f:\mathcal{\ri}^p\rightarrow\mathcal{\ri}^q$ with Grigoriev's flow
$w_f(u,v)$. Let $I$ (resp., $O$) denote the set of input (resp.,
output) vertices of $G$. Any post-dominator set $\Gamma$ for any
subset $I'\subseteq I$ with respect to any subset $O'\subseteq O$
satisfies $|\Gamma| \geq w_f(|I'|,|O'|)$.
\end{lemma}
\begin{proof}
Given $I'\subseteq I$ and $O'\subseteq O$, suppose the values of the
input variables corresponding to vertices in $I \setminus I'$ to be
fixed. Let $\Gamma$ be a post-dominator set for $I'$ with respect to
$O'$.\\ (i) According to the definition of flow of $f$, there exists an
assignment of the input variables corresponding to vertices in $I'$
such that the output variables in $O'$ can assume
$|\ri|^{w_{f}\left(|I'|,|O'|\right)}$ distinct values. (ii) As
there is no path from $I'$ to $O'$ which has not a vertex in $\Gamma$,
the values of the outputs in $O'$ are determined by the inputs in
$I\setminus I'$, which are fixed, and the values corresponding to the
vertices in $\Gamma$; hence the outputs in $O'$ can assume at most
$|\ri|^{|\Gamma|}$ distinct values. The claimed results follows
by a simple combination of observations (i) and (ii).
\end{proof}

\paragraph*{Properties of Strassen's Algorithm}
Consider Strassen's algorithm~\cite{strassen1969gaussian} when used to
compute $C=AB$, with $A$ and $B$ matrices $n\times n$ with entries
from the ring $\ri$ (see Algorithm~\ref{alg:strass} in
Appendix~\ref{app:proof_encoder} for a detailed presentation).
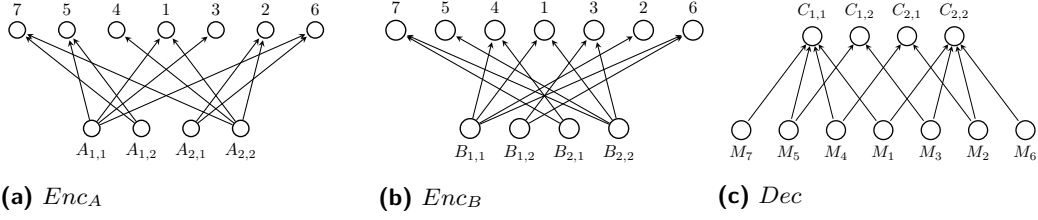
\begin{figure}
\begin{subfigure}{.31\textwidth}
  \centering
        \resizebox{\linewidth}{!}{
            \begin{tikzpicture}[scale = 0.5,
            > = stealth, 
            shorten > = 1pt, 
            auto,
            semithick 
        ]

        \tikzstyle{every state}=[
            draw = black,
            thick,
            minimum size = 1mm
        ]

        \node[state] at (0,0) (11) [label=below:$A_{1,1}$] {};
        \node[state] at (2,0) (12)  [label=below:$A_{1,2}$] {};
        \node[state] at (4,0) (21)  [label=below:$A_{2,1}$] {};
        \node[state] at (6,0) (22)  [label=below:$A_{2,2}$] {};
        \node[state] at (-3,4) (7)  [label=above:$7$] {};
        \node[state] at (-1,4) (5)  [label=above:$5$] {};
        \node[state] at (1,4) (4)  [label=above:$4$] {};
        \node[state] at (3,4) (1)  [label=above:$1$] {};
        \node[state] at (5,4) (3)  [label=above:$3$] {};
        \node[state] at (7,4) (2)  [label=above:$2$] {};
        \node[state] at (9,4) (6)  [label=above:$6$] {};

        \path[->] (11) edge  (5);
        \path[->] (11) edge  (1);
        \path[->] (11) edge  (3);
        \path[->] (11) edge  (6);
        
        \path[->] (12) edge  (7);
        \path[->] (12) edge  (5);
        
        \path[->] (21) edge  (2);
        \path[->] (21) edge  (6);
        
        \path[->] (22) edge  (2);
        \path[->] (22) edge  (1);
        \path[->] (22) edge  (4);
        \path[->] (22) edge  (7);
    \end{tikzpicture}
        }
  \caption{$Enc_A$}
  \label{fig:enca}
\end{subfigure}
\hspace{4mm}
\begin{subfigure}{.31\textwidth}
  \centering
   \resizebox{\linewidth}{!}{
  \begin{tikzpicture}[scale = 0.5,
            > = stealth, 
            shorten > = 1pt, 
            auto,
            node distance = 3cm, 
            semithick 
        ]

        \tikzstyle{every state}=[
            draw = black,
            thick,
            fill = white,
            minimum size = 4mm
        ]

        \node[state] at (0,0) (11) [label=below:$B_{1,1}$] {};
        \node[state] at (2,0) (12)  [label=below:$B_{1,2}$] {};
        \node[state] at (4,0) (21)  [label=below:$B_{2,1}$] {};
        \node[state] at (6,0) (22)  [label=below:$B_{2,2}$] {};
        \node[state] at (-3,4) (7)  [label=above:$7$] {};
        \node[state] at (-1,4) (5)  [label=above:$5$] {};
        \node[state] at (1,4) (4)  [label=above:$4$] {};
        \node[state] at (3,4) (1)  [label=above:$1$] {};
        \node[state] at (5,4) (3)  [label=above:$3$] {};
        \node[state] at (7,4) (2)  [label=above:$2$] {};
        \node[state] at (9,4) (6)  [label=above:$6$] {};

        \path[->] (11) edge  (4);
        \path[->] (11) edge  (1);
        \path[->] (11) edge  (2);
        \path[->] (11) edge  (6);
        
        \path[->] (12) edge  (3);
        \path[->] (12) edge  (6);
        
        \path[->] (21) edge  (7);
        \path[->] (21) edge  (4);
        
        \path[->] (22) edge  (7);
        \path[->] (22) edge  (1);
        \path[->] (22) edge  (5);
        \path[->] (22) edge  (3);
    \end{tikzpicture}
    }
  \caption{$Enc_B$}
  \label{fig:encb}
\end{subfigure}
\begin{subfigure}{.31\textwidth}
  \centering
     \resizebox{\linewidth}{!}{
  \begin{tikzpicture}[scale = 0.5,
            > = stealth, 
            shorten > = 1pt, 
            auto,
            node distance = 3cm, 
            semithick 
        ]

        \tikzstyle{every state}=[
            draw = black,
            thick,
            fill = white,
            minimum size = 4mm
        ]

        \node[state] at (0,4) (C11) [label=above:$C_{1,1}$] {};
        \node[state] at (2,4) (C12)  [label=above:$C_{1,2}$] {};
        \node[state] at (4,4) (C21)  [label=above:$C_{2,1}$] {};
        \node[state] at (6,4) (C22)  [label=above:$C_{2,2}$] {};
        \node[state] at (-3,0) (7)  [label=below:$M_7$] {};
        \node[state] at (-1,0) (5)  [label=below:$M_5$] {};
        \node[state] at (1,0) (4)  [label=below:$M_4$] {};
        \node[state] at (3,0) (1)  [label=below:$M_1$] {};
        \node[state] at (5,0) (3)  [label=below:$M_3$] {};
        \node[state] at (7,0) (2)  [label=below:$M_2$] {};
        \node[state] at (9,0) (6)  [label=below:$M_6$] {};

        \path[<-] (C11) edge  (4);
        \path[<-] (C11) edge  (1);
        \path[<-] (C11) edge  (7);
        \path[<-] (C11) edge  (5);
        
        \path[<-] (C12) edge  (3);
        \path[<-] (C12) edge  (5);
        
        \path[<-] (C21) edge  (2);
        \path[<-] (C21) edge  (4);
        
        \path[<-] (C22) edge  (6);
        \path[<-] (C22) edge  (1);
        \path[<-] (C22) edge  (2);
        \path[<-] (C22) edge  (3);
    \end{tikzpicture}
    }
  \caption{$Dec$}
  \label{fig:dec}
\end{subfigure}
\caption[Basic constructing blocks of Strassen's algorithm CDAG]{Basic constructing blocks of Strassen's CDAG. Note that $Enc_A$ and $Enc_B$ are isomorphic.}
\label{fig:fig}
\end{figure}
Let $H^{n \times n}$ denote the corresponding CDAG. For $n\geq 2$, $H^{n \times n}$ can be obtained by using a recursive construction which mirrors the recursive structure of the algorithm. The base of the construction is the $H^{2 \times 2}$ CDAG which corresponds to the multiplication of two $2 \times 2$ matrices using Strassen's algorithm (Figure~\ref{fig:bsestra}). $H^{2n \times 2n}$ can then be constructed by composing seven copies of $H^{n \times n}$, each corresponding to one of the seven sub-products generated by the algorithm (see Figure~\ref{fig:strarec}):  
  $n^2$ vertex-disjoint copies of CDAGs $Enc_A$ (resp., $Enc_B$) are used to connect the input vertices of $H^{2n \times 2n}$ which correspond to the values of the input matrix $A$ (resp., $B$) to the appropriate input vertices of the seven sub-CDAGs $H^{n \times n}_i$; the output vertices sub-CDAGs $H^{n \times n}_i$ (which  correspond to outputs of the seven sub-products) are connected to the opportune output vertices of the entire $H^{2n \times 2n}$ CDAG  using $n^2$ copies of the decoder sub-CDAG $Dec$.
  
  \begin{figure}[bt]
  \begin{subfigure}{.48\textwidth}
  \centering
     \resizebox{\linewidth}{!}{
     \begin{tikzpicture}[scale = 0.5,
            > = stealth, 
            shorten > = 1pt, 
            auto,
            node distance = 3cm, 
            semithick 
        ]

        \tikzstyle{every state}=[
            draw = black,
            thick,
            minimum size = 4mm
        ]
		\node[state] at (0,0) (A11) [label=below:$A_{1,1}$] {};
        \node[state] at (2,0) (A12)  [label=below:$A_{1,2}$] {};
        \node[state] at (4,0) (A21)  [label=below:$A_{2,1}$] {};
        \node[state] at (6,0) (A22)  [label=below:$A_{2,2}$] {};
        \node[state] at (-3,4) (A7)  [fill = blue] {};
        \node[state] at (-1,4) (A5)  [fill = blue] {};
        \node[state] at (1,4) (A4)    {};
        \node[state] at (3,4) (A1)   [fill = blue]{};
        \node[state] at (5,4) (A3)    {};
        \node[state] at (7,4) (A2)   [fill = blue] {};
        \node[state] at (9,4) (A6)   [fill = blue] {};

        \path[->] (A11) edge  (A5);
        \path[->] (A11) edge  (A1);
        \path[->] (A11) edge  (A3);
        \path[->] (A11) edge  (A6);
        
        \path[->] (A12) edge  (A7);
        \path[->] (A12) edge  (A5);
        
        \path[->] (A21) edge  (A2);
        \path[->] (A21) edge  (A6);
        
        \path[->] (A22) edge  (A2);
        \path[->] (A22) edge  (A1);
        \path[->] (A22) edge  (A4);
        \path[->] (A22) edge  (A7);
        
        \node[state] at (14,0) (B11) [label=below:$B_{1,1}$] {};
        \node[state] at (16,0) (B12)  [label=below:$B_{1,2}$] {};
        \node[state] at (18,0) (B21)  [label=below:$B_{2,1}$] {};
        \node[state] at (20,0) (B22)  [label=below:$B_{2,2}$] {};
        \node[state] at (11,4) (B7)  [fill = blue] {};
        \node[state] at (13,4) (B5)  {};
        \node[state] at (15,4) (B4)  [fill = blue] {};
        \node[state] at (17,4) (B1)  [fill = blue] {};
        \node[state] at (19,4) (B3)  [fill = blue] {};
        \node[state] at (21,4) (B2)   {};
        \node[state] at (23,4) (B6)  [fill = blue] {};

        \path[->] (B11) edge  (B4);
        \path[->] (B11) edge  (B1);
        \path[->] (B11) edge  (B2);
        \path[->] (B11) edge  (B6);
        
        \path[->] (B12) edge  (B3);
        \path[->] (B12) edge  (B6);
        
        \path[->] (B21) edge  (B7);
        \path[->] (B21) edge  (B4);
        
        \path[->] (B22) edge  (B7);
        \path[->] (B22) edge  (B1);
        \path[->] (B22) edge  (B5);
        \path[->] (B22) edge  (B3);

        \node[state] at (7,12) (C11) [label=above:$C_{1,1}$] {};
        \node[state] at (9,12) (C12)  [label=above:$C_{1,2}$] {};
        \node[state] at (11,12) (C21)  [label=above:$C_{2,1}$] {};
        \node[state] at (13,12) (C22)  [label=above:$C_{2,2}$] {};
        \node[state] at (2.5,8) (7)  [fill = red, label=left:$M_7$] {};
        \node[state] at (5,8) (5)  [fill = red, label=left:$M_5$] {};
        \node[state] at (7.5,8) (4)  [fill = red, label=left:$M_4$] {};
        \node[state] at (10,8) (1)  [fill = red, label=left:$M_1$] {};
        \node[state] at (12.5,8) (3)  [fill = red, label=left:$M_3$] {};
        \node[state] at (15,8) (2)  [fill = red, label=left:$M_2$] {};
        \node[state] at (17.5,8) (6)  [fill = red, label=left:$M_6$] {};
        
        \node[state, draw= white] at (7,1.5) (ena)  [label=right: \LARGE $Enc_A$] {};
        \node[state, draw= white] at (20,1.5) (enb)  [label=right: \LARGE $Enc_B$] {};
        \node[state, draw= white] at (0,10.5) (dec)  [label=right: \LARGE$Dec$] {};

        \path[<-] (C11) edge  (4);
        \path[<-] (C11) edge  (1);
        \path[<-] (C11) edge  (7);
        \path[<-] (C11) edge  (5);
        
        \path[<-] (C12) edge  (3);
        \path[<-] (C12) edge  (5);
        
        \path[<-] (C21) edge  (2);
        \path[<-] (C21) edge  (4);
        
        \path[<-] (C22) edge  (6);
        \path[<-] (C22) edge  (1);
        \path[<-] (C22) edge  (2);
        \path[<-] (C22) edge  (3);
        
        \path[->] (A7) edge  (7);
        \path[->] (B7) edge  (7);
        \path[->] (A5) edge  (5);
        \path[->] (B5) edge  (5);
        \path[->] (A4) edge  (4);
        \path[->] (B4) edge  (4);
        \path[->] (A1) edge  (1);
        \path[->] (B1) edge  (1);
        \path[->] (A3) edge  (3);
        \path[->] (B3) edge  (3);
        \path[->] (A2) edge  (2);
        \path[->] (B2) edge  (2);
        \path[->] (A6) edge  (6);
        \path[->] (B6) edge  (6);
    \end{tikzpicture}
     }
     \caption{Strassen's $H^{2\times 2}$ CDAG}
     \label{fig:bsestra}
  \end{subfigure}	
  \hspace{1mm}
  \begin{subfigure}{.49\textwidth}
  \centering
     \resizebox{\linewidth}{!}{
     	\begin{tikzpicture}[
            > = stealth, 
            shorten > = 1pt, 
            auto,
            node distance = 2cm, 
            semithick 
        ]

        \tikzstyle{every state}=[
            draw = black,
            thick,
            minimum size = 8mm
        ]
		\node[state, tokens=4] at (0,0) (A11) [label=below: \huge $A_{1,1}$] {};
        \node[state, tokens=4] at (2,0) (A12)  [label=below:\huge $A_{1,2}$] {};
        \node[state, tokens=4] at (4,0) (A21)  [label=below:\huge $A_{2,1}$] {};
        \node[state, tokens=4] at (6,0) (A22)  [label=below:\huge $A_{2,2}$] {};
        \node[state, draw = blue, tokens=4] at (-3,4) (A7)   {};
        \node[state, draw = blue, tokens=4] at (-1,4) (A5)   {};
        \node[state, tokens=4] at (1,4) (A4)    {};
        \node[state, draw = blue, tokens=4] at (3,4) (A1)    {};
        \node[state, tokens=4] at (5,4) (A3)    {};
        \node[state, draw = blue, tokens=4] at (7,4) (A2)    {};
        \node[state, draw = blue, tokens=4] at (9,4) (A6)    {};

        \path[->] (A11) edge  (A5);
        \path[->] (A11) edge  (A1);
        \path[->] (A11) edge  (A3);
        \path[->] (A11) edge  (A6);
        
        \path[->] (A12) edge  (A7);
        \path[->] (A12) edge  (A5);
        
        \path[->] (A21) edge  (A2);
        \path[->] (A21) edge  (A6);
        
        \path[->] (A22) edge  (A2);
        \path[->] (A22) edge  (A1);
        \path[->] (A22) edge  (A4);
        \path[->] (A22) edge  (A7);
        
        \node[state, tokens=4] at (14,0) (B11) [label=below:\huge $B_{1,1}$] {};
        \node[state, tokens=4] at (16,0) (B12)  [label=below:\huge $B_{1,2}$] {};
        \node[state, tokens=4] at (18,0) (B21)  [label=below:\huge $B_{2,1}$] {};
        \node[state, tokens=4] at (20,0) (B22)  [label=below:\huge $B_{2,2}$] {};
        \node[state, draw = blue, tokens=4] at (11,4) (B7)   {};
        \node[state, tokens=4] at (13,4) (B5)  {};
        \node[state, draw = blue, tokens=4] at (15,4) (B4)   {};
        \node[state, draw = blue, tokens=4] at (17,4) (B1)  {};
        \node[state, draw = blue, tokens=4] at (19,4) (B3)   {};
        \node[state, tokens=4] at (21,4) (B2)   {};
        \node[state, draw = blue, tokens=4] at (23,4) (B6)   {};

        \path[->] (B11) edge  (B4);
        \path[->] (B11) edge  (B1);
        \path[->] (B11) edge  (B2);
        \path[->] (B11) edge  (B6);
        
        \path[->] (B12) edge  (B3);
        \path[->] (B12) edge  (B6);
        
        \path[->] (B21) edge  (B7);
        \path[->] (B21) edge  (B4);
        
        \path[->] (B22) edge  (B7);
        \path[->] (B22) edge  (B1);
        \path[->] (B22) edge  (B5);
        \path[->] (B22) edge  (B3);

        \node[state, tokens=4] at (7,12) (C11) [label=above:\huge $C_{1,1}$] {};
        \node[state, tokens=4] at (9,12) (C12)  [label=above:\huge $C_{1,2}$] {};
        \node[state, tokens=4] at (11,12) (C21)  [label=above:\huge $C_{2,1}$] {};
        \node[state, tokens=4] at (13,12) (C22)  [label=above:\huge $C_{2,2}$] {};
        \node[state, draw = red] at (1,8) (7)  {\huge $H^{n\times n}_7$};
        \node[state, draw = red] at (4,8) (5)  {\huge $H^{n\times n}_5$};
        \node[state, draw = red] at (7,8) (4)  {\huge$H^{n\times n}_4$};
        \node[state, draw = red] at (10,8) (1)  {\huge $H^{n\times n}_1$};
        \node[state, draw = red] at (13,8) (3) {\huge $H^{n\times n}_3$};
        \node[state, draw = red] at (16,8) (2)  {\huge $H^{n\times n}_2$};
        \node[state, draw = red] at (19,8) (6)  {\huge $H^{n\times n}_6$};
        
        \node[state, draw= white] at (7,1.5) (ena)  [label=right: \Huge $n^2\times Enc_A$] {};
        \node[state, draw= white] at (19.5,1.5) (enb)  [label=right: \Huge $n^2 \times Enc_B$] {};
        \node[state, draw= white] at (-1,10.5) (dec)  [label=right:\Huge $n^2 \times Dec$] {};

        \path[<-] (C11) edge  (4);
        \path[<-] (C11) edge  (1);
        \path[<-] (C11) edge  (7);
        \path[<-] (C11) edge  (5);
        
        \path[<-] (C12) edge  (3);
        \path[<-] (C12) edge  (5);
        
        \path[<-] (C21) edge  (2);
        \path[<-] (C21) edge  (4);
        
        \path[<-] (C22) edge  (6);
        \path[<-] (C22) edge  (1);
        \path[<-] (C22) edge  (2);
        \path[<-] (C22) edge  (3);
        
        \path[->] (A7) edge  (7);
        \path[->] (B7) edge  (7);
        \path[->] (A5) edge  (5);
        \path[->] (B5) edge  (5);
        \path[->] (A4) edge  (4);
        \path[->] (B4) edge  (4);
        \path[->] (A1) edge  (1);
        \path[->] (B1) edge  (1);
        \path[->] (A3) edge  (3);
        \path[->] (B3) edge  (3);
        \path[->] (A2) edge  (2);
        \path[->] (B2) edge  (2);
        \path[->] (A6) edge  (6);
        \path[->] (B6) edge  (6);
    \end{tikzpicture}
     }
     \caption{Recursive construction of the $H^{2n\times 2n}$ CDAG}
     \label{fig:strarec}
  \end{subfigure}

  \caption{Blue vertices represent combinations of the input values from the factor matrices $A$ and $B$ which are used as input values for the sub-problems $M_i$; red vertices represent the output of the seven sub-problems which are used to compute the output values of the product matrix $C$.}
\end{figure}
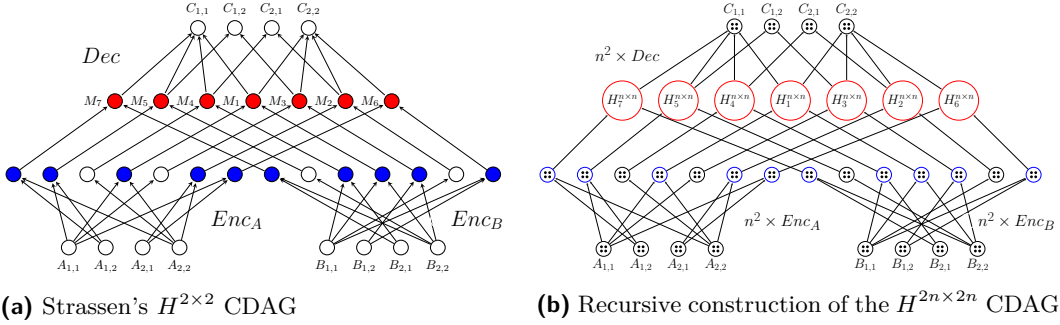 

In our proofs we will leverage the following two properties of Strassen's CDAG:
\begin{lemma}\label{lem:disictness_sub_cdag}
Let $H^{n \times n}$  denote the CDAG of Strassen's algorithm for input matrices of size $n \times n$. For $0 \leq i \leq \log n-1$, there are exactly $7^i$ sub-CDAGs $H^{n/2^i \times n/2^i}$ which do not share any vertex in $H^{n \times n}$ (i.e., they are \emph{vertex disjoint sub-CDAGs of $H^{n \times n}$}).
\end{lemma}
\begin{lemma}\label{lem:conneconder}
Given an encoder CDAG, for any subset $Y$ of its output vertices, there exists a sub-set $X$ of its input vertices such that each vertex in $X$ can be connected to a distinct vertex in $Y$ where $\min\{|Y|, 1 + \lceil\left(|Y|-1\right)/2\rceil \}\leq |X|\leq |Y|$.
\end{lemma}

The validity of this lemmas can be verified through an analysis of respectively, the recursive structure of Strassen's CDAG and the properties of  the encoding sub-CDAGs $Enc_A$ and $Enc_B$. We refer the reader to Appendix~\ref{app:proof_encoder} for detailed proofs.

\paragraph*{Strassen-like algorithms}
A $(n_0,m_0)$-Strassen-like algorithm has a recursive structure for
which in the ``\emph{base case}'' two $n_0\times n_0$ matrices are
multiplied using $m_0$ scalar multiplications, whose result are then
combined to obtain the product matrix.  Given input matrices of size
$n\times n$, the algorithm splits them into $n_0^2$ sub-matrices of
size $n/n_0 \times n/n_0$ and then proceeds block-wise, according to
the base case. Additions (resp., subtractions) in the base case are
interpreted as additions (resp., subtractions) of blocks and are
performed element-wise. Multiplications in the base case are
interpreted as multiplications of blocks and are executed by
recursively calling the algorithm. In our analysis we consider
Strassen-like algorithms for which each linear combination of the
input sub-matrices is used in only one multiplication.

Let $\stral{n\times n}$ be the CDAG corresponding to a $\left(n_0,
m_0\right)$-Strassen-like algorithm for input matrices of size
$n\times n$ . $\stral{n\times n}$ has a recursive structure analogous
to that of Strassen's CDAG $H^{n\times n}$. The base element
$\stral{n_0 \times n_0}$ of the CDAG construction corresponding to the
base case consists of two \emph{encoding graphs} (corresponding
respectively to the sub-CDAGs $Enc_A$ and $Enc_B$ for $H^{2\times
  2}$), which compute $m_0$ linear combinations of entries of
respectively the factor matrix $A$ and of $B$. Corresponding pairs are
then multiplied and the outputs are then combined using an
\emph{decoder} sub-CDAG to obtain the output.

Although in general the sub-problems generated by a Strassen-like
algorithm may not be input disjoint, in~\cite{scott2015matrix} Scott
et al. showed that a property similar to
Lemma~\ref{lem:disictness_sub_cdag} holds:

\begin{lemma}[Lemma 1~\cite{scott2015matrix}]\label{lem:vdigraphslike}
Let $\stral{n \times n}$ denote the CDAG of a
$(n_0,m_0)$-Strassen-like algorithm for input matrices of size $n
\times n$. For $0 \leq i \leq \log n-1$, there are at least
$m_0^{i-2}$ vertex-disjoint sub-CDAGs $\stral{n/n_0^i \times n/n_0^i}$
in $\stral{n\times n}$.
\end{lemma}
For completeness, we summarize their result in
Appendix~\ref{app:proof_encoder}.

\section{Lower bounds for schedules without recomputation}\label{sect:stranore}
In our presentation we denote as $G^{n \times n}$ the CDAG corresponding to the execution of an \emph{unspecified} algorithm which implements the square matrix multiplication function $f_{n\times n}:\ri^{2n^2}\rightarrow\ri^{n^2}$. 
Let $G^{q,n\times n}$ be a CDAG composed by $q$ vertex-disjoint CDAGs
$G^{n\times n}$. The set $I$ (resp., $O$) of input (resp., output)
vertices of $G^{q,n\times n}$ is given by the union of the sets of the
input (resp., output) vertices of the $q$ sub-CDAGs $G^{n\times
  n}$. According to this definition, the CDAG of Strassen $H^{n\times
  n}$ is a $G^{7^i,n/2^i\times n/2^i}$ CDAG for
$i\in{0,1,\ldots,\log_2 n -1 }$ (Lemma~\ref{lem:disictness_sub_cdag}). 

\begin{lemma}\label{lem:newbasenr}
Given $G^{q,n\times n}$, let $O'\in O$ be a subset of its the output
vertices. For any subset $\Gamma$ of the vertices of $G^{q,n\times n}$
with $|\Gamma|\leq 2|O'|$, the set $I'\subseteq I$ of the input
vertices which are not post-dominated by $\Gamma$ satisfies $|I'|\geq
2n\sqrt{|O'|-2|\Gamma|}$.
\end{lemma}
\begin{proof}
	Let $O'_j$ (resp., $\Gamma_j$) denote the subset of $O'$ (resp., $\Gamma$) which correspond to vertices in the $j$-th sub CDAG $G^{n\times n}_j$ for $j\in\{1,2,\ldots,q\}$. As, by hypothesis, the sub-CDAGs $G^{n\times n}_q$ are vertex disjoint, $O'_1,O'_2,\ldots,O'_q$ (resp., $\Gamma_1,\Gamma_2,\ldots,\Gamma_q$) constitute a partition of $O'$ (resp., $\Gamma$). 
	
	Let $I''_j\subseteq \mathcal{X}$ denote the set of input vertices of $G^{n\times n}_j$ for which $\Gamma_j$ is a post-dominator with respect to $O_j$. 
	From Lemma~\ref{lem:info_flo_mat_mul} and Lemma~\ref{lem:flowpost} the following condition must hold:
	\begin{equation}\label{eq:newlem1}
		|\Gamma_j|\geq w_{n\times n} \geq 
		\frac{1}{2}\left(|O'_j|-\frac{\left(2n^2 - |I''_j|\right)^2}{4n^2}\right).
	\end{equation}	
	Let $I'_j =I_j \setminus I''_j$ denote the set of input vertices of $G^{n\times n}_j$ for which $\Gamma_j$ is a post-dominator with respect to $O_j$. Since $|I|= 2n^2$, from (\ref{eq:newlem1}) we have $|I'_j|^2  \geq 4n^2\left(|O'_j|-2|\Gamma_j|\right)$.
	
	As the sub-CDAGs $G^{n\times n}_j$ are vertex disjoint, we can state that all the vertices in $I'=\cup_{j=1}^i I'_j$ are not post-dominated by $\Gamma$ and
	$
		|I'|^2 = \sum_{j=1}^q|I'_j|^2 \geq   4n^2\sum_{j=1}^q\left(|O'_j|-2|\Gamma_j|\right) = 4n^2\left(|O'|-2|\Gamma|\right)
	$.
\end{proof}

\begin{corollary}\label{cor:domnr}
	Given $G^{q,n\times n}$, the minimum size of a dominator set of any subset of the output vertices $O'\in O$ is at least $\lceil|O'|/2\rceil$.	
\end{corollary}

These results allows us to obtain a general theorem on the \io
complexity of matrix multiplication algorithms under the no-recomputation assumption.

\begin{theorem}[Lower bound \io complexity  matrix multiplication with no recomputation]
\label{thm:strassnr}
Let $G^{\alg}$  the CDAG corresponding to an algorithm $\alg$ which computes the product of two square matrices $A,B\in \ri^{n^2}$. Suppose that  $G^{\alg}$ has $q$ vertex disjoint sub-CDAGs $G^{2\sqrt{M}\times 2\sqrt{M}}$ each of which corresponds to input-disjoint sub-products of matrices of size $2\sqrt{M}\times 2\sqrt{M}$ generated by $\alg$. 
Assuming no intermediate result is ever computed more than once, the \io-complexity of $\alg$ when run on a sequential machine with a cache of size $M$ is:
\begin{equation}\label{eq:seqnewnr}
	IO_{G^{\alg}}(M) \geq qM.
\end{equation}
If run on $P$ processor each equipped with a local memory of size M the \io complexity is:
\begin{equation}\label{eq:parnewnr}
	IO_{G^{\alg}}(P,M) \geq qM/P.
\end{equation}
\end{theorem}
\begin{proof}
	As first observation, note that the $q$ vertex-disjoint sub-CDAGs $G^{2\sqrt{M}\times 2\sqrt{M}\times n}$ constitute a $G^{q,2\sqrt{M}\times 2\sqrt{M}}$ CDAG as previously defined. The set $X$ of input (resp. $Y$ of output) vertices of $G^{q,2\sqrt{M}\times 2\sqrt{M}}$ is composed by the union of the input (resp., output) vertices of the $q$ sub-CDAGs $G^{2\sqrt{M}\times 2\sqrt{M}}$. Clearly, $|X|=|Y|=q4M$.
	
	We start by proving (\ref{eq:seqnewnr}). Let $\mathcal{C}$ be any computational schedule for the sequential execution of the algorithm $\alg$ using a cache of size $M$ for which each intermediate value is computed just once. Each intermediate value must therefore be kept in memory (either cache or slow) until all operations using it as an operand have been computed. The operations corresponding to vertices in $G^{\alg}$ are executed according to a topological ordering of the vertices, that is, any value corresponding to a vertex in $u\in X$ which can be connected through a directed path to a vertex in $v\in Y$ will be computed in $\mathcal{C}$ before $v$. We partition $\mathcal{C}$ into $q$ segments such that during each segment $\mathcal{C}_j$ the values corresponding to exactly $4M$ vertices in $Y$ (denoted as $Y_j$) are computed for the first (and only) time.	
	We shall now verify that at least $M$ \io operations are to be executed during every segment $\mathcal{C}_j$. The proof is by contradiction.   
	 Let $\Gamma_j$ denote the set of vertices of $G^{\alg}$ which correspond to the at most $M$ values which are  stored in the cache at the beginning of the segment and the at most $M-1$ values which are loaded into the cache form the slow memory during $\mathcal{C}_j$ by means of a \emph{read} \io operation. We thus have $|\Gamma_j|\leq 2M -1$.
	  In order for the $4M$ values corresponding to vertices in $Y_j$ to be computed during the segment without any additional \io operation, there must be no path connecting any vertex in $Y_i$ to any vertex in $I$ which does not have at least one vertex in $\Gamma_j$ (i.e. $\Gamma_j$ has to be a \emph{dominator set} of $Y_j$). If any such path exists then there is a previously computed intermediate result $v$ which is required for the computation of at lest one of the values corresponding to a vertex in $Y_j$ which is neither residing in the cache at the beginning of $\mathcal{C}_j$, nor loaded from slow memory to the cache. As no intermediate result can be computed twice this would lead to a contradiction.  From Corollary~\ref{cor:domnr}, we have that any sub-set of $4M$ elements of $X$ has dominator size at least $2M$. This leads to a contradiction.
	
	At least $M$ \io operations are thus executed during each segment $\mathcal{C}_j$. Since, by construction, the $q$ segments are not overlapping, we can conclude that at least $qM$ \io operations are necessary for the execution of $\mathcal{C}$. This concludes the proof for lower bound in (\ref{eq:seqnewnr}).
	
	The proof for the lower bound for the parallel model in equation (\ref{eq:parnewnr}), follows a similar strategy. If $P$ processors are being used, at least one such processor $P^*$ must compute at least $|Y|/P = q4M/P$ values corresponding to vertices in $Y$. The bound follows by applying the same argument discussed for the sequential case to the computation executed by $P^*$.
	\end{proof}

This general result can be steadily applied to Strassen-like algorithms.
\begin{corollary}
Consider Strassen's algorithm being used to multiply two square
matrices $A,B\in\ri^{n\times n}$.  Assuming no intermediate result is
ever computed more than once, the I/O-complexity of the algorithm run
on a sequential machine with a cache of size $M$ is:
\begin{equation}\label{eq:stranr}
	IO_{H^{n\times n}}(M) \geq \frac{1}{7}  \left( \frac{n}{\sqrt{M}}\right)^{\log_2 7} M.
\end{equation}
If run on $P$ processor each equipped with local memory of size $M\leq n^2$ the \io complexity is:
\begin{equation}\label{eq:stranrpar}
	IO_{H^{n\times n}}(P,M) \geq \frac{1}{7}  \left( \frac{n}{\sqrt{M}}\right)^{\log_2 7} \frac{M}{P}.
\end{equation}
Additionally, the I/O-complexity of any $\left(n_0,m_0\right)$-Strassen-like algorithm run on a sequential machine with a cache of size $M$ is:
\begin{equation}\label{eq:stranrlike}
	IO_{\stral{n\times n}}(M) = \BOme{  \left( \frac{n}{\sqrt{M}}\right)^{\log_{n_0} m_0} M}.
\end{equation}
If run on $P$ processor each equipped with local memory of size $M\leq n^2$ the \io complexity is:
\begin{equation}\label{eq:stranrlikepar}
	IO_{\stral{n\times n}}(P,M) =\BOme{  \left( \frac{n}{\sqrt{M}}\right)^{\log_{n_0} m_0} \frac{M}{P}}.
\end{equation}
\end{corollary}
\begin{proof}
	We provide a simple proof for the sequential cases in (\ref{eq:stranr}) and (\ref{eq:stranrlike}). Let us assume without loss of generality that $n = 2^a$ and $\sqrt{M}=2^b$  for some $a,b\in \mathbb{N}$. At least $3M$ \io operations are necessary in order to read all the $2n^2$ input values form slow memory to the cache and to write the $n^2$ output values to the slow memory once they have been computed. The statement of the theorem is therefore trivially verified if $n\leq 2\sqrt{M}$. 
For $n\geq 2\sqrt{M}$, the result in (\ref{eq:stranr}) follows from applying Theorem~\ref{thm:strassnr} to the $H^{n\times n}$ CDAG which, from Lemma~\ref{lem:disictness_sub_cdag}, has $\left(n/2\sqrt{M}\right)^{\log_2 7}$ vertex-disjoint sub CDAGs $H^{2\sqrt{M} \times 2\sqrt{M}}$. 
The result in (\ref{eq:stranrlike}) similarly follows from applying Lemma~\ref{lem:vdigraphslike}. The results for the parallel model in (\ref{eq:stranrpar}) and (\ref{eq:stranrlikepar}) can be obtained using a generalization similar to the one described in Theorem~\ref{thm:strassnr}.
\end{proof}
While our lower bounds correspond asymptotically to the known bounds
in~\cite{scott2015matrix}, our technique yields a simpler analysis,
especially for the \io complexity of Strassen-like algorithms. Our
proof technique is based on the analysis of the recursive structure of
Strassen-like algorithms and on the identification of the
vertex-disjoint sub-CDAGs corresponding to the various
sub-problems. The fact that Theorem~\ref{thm:strassnr} applies for
\emph{any} square matrix multiplication algorithm, allows us to obtain
significant bounds without a detailed analysis of the CDAG
corresponding to the \emph{specific} algorithm being considered.

This property further suggests that our technique may be amenable to
deal with hybrid ``\emph{non-stationary}'' multiplication algorithms, which allow mixing of schemes of the previous Strassen-like class in different recursive levels (for the ``\emph{uniform}'' sub-class)~\cite{douglas1994gemmw}, or even within the same level (for the ``\emph{non-uniform}'' sub-class).

\section{Lower bounds for schedules with recomputation}\label{sec:stragen}
Under no-recomputation assumption once an input value is loaded in memory or an intermediate result is calculated it is then necessary to maintain it in memory (either cache or slow) until the result of each operation which uses it as an input argument has been evaluated. This constitutes the foundation of the lower bound technique discussed in Theorem~\ref{thm:strassnr} as well as several other techniques discussed in literature (among others, the ``\emph{dichotomy width} technique''~\cite{bilardi1999processor}, the ``\emph{boundary flow} technique''~\cite{ranjan2012upper}), including those yielding \io lower bound for Strassen's algorithm~\cite{ballard2012graphrec,scott2015matrix}. If recomputation is allowed, intermediate results can instead be deleted from \emph{all} memory and recomputed. This introduces a substantial complication in the theoretical analysis of the \io cost (see~\cite{ballard2011minimizing} for an extensive discussion). \\

In this section we present a new lower bound technique which yields a novel asymptotically tight \io lower bound for Strassen's algorithm both in the sequential and parallel model.
We start by presenting some technical lemmas which  will then be used for the proof of our main result in Theorem~\ref{thm:genstrass}. 

\begin{lemma}\label{lem:stra_part1}
Let $H^{n\times n}$ be the CDAG which corresponds to the execution of Strassen's matrix multiplication algorithm for input matrices $A,B\in \ri^{n\times n}$, with $n\geq 2\sqrt{M}$. Let $\mathcal{Y}$ (resp., $\mathcal{Z}$) denote the set of input (resp., output) vertices of the $\left(n/2\sqrt{M}\right)^{\log_2 7}$ sub-CDAGs $H^{2\sqrt{M}\times 2\sqrt{M}}$. Furthermore, let $\Gamma$ be a subset of the set of \emph{internal} (i.e., not input) vertices of the sub-CDAGs $H^{2\sqrt{M}\times 2\sqrt{M}}$. Let $\mathcal{X}$ denote the set on input vertices of $H^{n\times n}$. For any subset $Z\subseteq \mathcal{Z}$ such that $|Z|\geq 2|\Gamma|$ there exists a set $X\subset\mathcal{X}$ with $|X|\geq 4\sqrt{M\left(|Z|-2|\Gamma|\right)}$ such that  vertices in $X$ can be connected to vertices in a subset $Y\in \mathcal{Y}$, with $|X|=|Y|$ via vertex disjoint paths. Furthermore all vertices in $Y$ can be connected to a vertex in $Z$ by a directed path which does not include any vertex in $\Gamma$.
\end{lemma}
\begin{proof}
	The proof is by induction on the size of the input matrices.\\
\emph{Base:} In the base case we have $n = 2\sqrt{M}$. We therefore
have $H^{n\times n} = H^{2\sqrt{M}\times 2\sqrt{M}}$ and the sets
$\mathcal{Y}$ and $\mathcal{X}$ coincide. We can verify that the
statement holds by simply applying Lemma~\ref{lem:newbasenr} as
$H^{2\sqrt{M}\times 2\sqrt{M}}$ corresponds to a $G^{1,2\sqrt{M}\times
  2\sqrt{M}}$ CDAG.\\
\emph{Inductive step:} Let us now assume that the statement is
verified for $H^{n\times n}$, with $n\geq 2\sqrt{M}$. We shall show
that the statement is verified for $H^{2n\times 2n}$ as well.  Let
$H^{n\times n}_1,H^{n\times n}_2,\ldots,H^{n\times n}_7$ denote the
seven sub-CDAGs of $H^{2n\times 2n}$, each corresponding to one of the
seven sub-products generated by the first recursive step of Strassen's
algorithm. Let $Z_i$ (resp., $\mathcal{Y}_i$) denote the subset of $Z$
(resp., $\mathcal{Y}$) which correspond to vertices in $H^{n\times
  n}_i$. As, from Lemma~\ref{lem:disictness_sub_cdag}, the seven
sub-CDAGs $H^{n\times n}_i$ are vertex disjoint among themselves,
$Z_1,Z_2,\ldots,Z_7$ (resp., $\mathcal{Y}_1,\mathcal{Y}_2,\ldots,
\mathcal{Y}_7$) constitute a partition of $Z$ (resp., $\mathcal{Y}$).

For $i\in \{1,2,\ldots,7\}$, let $\Gamma_i$ be the subset of vertices in $\Gamma$ in the sub-CDAGs $H^{n\times n}_i$. Again, as the seven sub-CDAGs $H^{n\times n}_i$ are vertex disjoint among themselves, we have that $\Gamma_1,\Gamma_2,\ldots,\Gamma_7$ is a partition of $\Gamma$.  This implies $\sum_{i=1}^7 |\Gamma_i| = |\Gamma|$. Let $\delta_i = \max\{0, |Z_i|-2|\Gamma_i|\}$, we have $\delta = \sum_{i=1}^7 \delta_i \geq |Z| - 2|\Gamma|$. 
	
	Applying the inductive hypothesis to each sub-CDAG $H^{n\times n}_i$, we have that there is a subset $Y_i\subseteq \mathcal{Y}_i$ with $|Y_i|\geq 4\sqrt{M\delta_i}$ such that vertices of $Y_i$ are connected  to vertices in $Z_i$ via paths which do not include any vertex in $\Gamma_i$. Furthermore vertices in $Y_i$ can be connected to a subset $K_i$ of the input vertices of $H^{n\times n}_i$ with $|K_i|=|Y_i|$, using vertex-disjoint paths. Since the sub-CDAGs $H^{n\times n}_i$ are vertex disjoint, so are the paths connecting vertices in $Y_i$ to vertices in $K_i$. In order to conclude our proof we need to show that is possible to extend at least $4\sqrt{M\left(|Z|-2|\Gamma\|\right)}$ of these paths to vertices in $\mathcal{X}$ while still being vertex disjoint. 
	
	According to the construction of Strassen's CDAG, vertices in $\mathcal{X}$  are connected to vertices in $K_1,K_2,\ldots,K_7$ by means of $2n^2$ encoding sub-CDAGs $Enc_A$ and $Enc_B$ ($n^2$ of each). For the construction of $H^{2n \times 2n}$, none of these encoding sub-CDAGs shares any input or output vertices. For any given encoder sub-CDAGs each of its output vertices belongs to a different sub-CDAG $H^{n\times n}_i$. This fact ensures that for a single sub-CDAG $H^{n\times n}_i$ it is possible to connect all the vertices in $K_i$ to a subset of the vertices in $\mathcal{X}$ via vertex disjoint paths.
	
	For each of the $2n^2$ encoder sub-CDAGs, let us consider the vector $\mathbf{y}_j\in\{0,1\}^7$. We have that $\mathbf{y}_j[i] = 1$ if the corresponding $i$-th output vertex (respectively according to the numbering indicated in Figure~\ref{fig:enca} or Figure~\ref{fig:encb}) is in $K_i$ or $\mathbf{y}_j[i] = 0$ otherwise. We therefore have that $|\mathbf{y}_j|$ corresponds to the number of output vertices of the $j$-th encoder sub-CDAG which are in $K$.
	 From Lemma~\ref{lem:conneconder}, we have that for each encoder sub-CDAG there exists a subset $X_j\in\mathcal{X}$ of the input vertices of the $j$-th encoder sub-CDAG for which is possible to connect each vertex in $X_j$ to a distinct output vertex of the $j$-th encoder sub-CDAG using vertex disjoint paths, each constituted by a singular edge with $\min\{|\mathbf{y}_j|, 1 +\lceil\left(|\mathbf{y}_j|-1\right)/2\rceil\}\leq |X_j|\leq |\mathbf{y}_j|$. The number of vertex disjoint paths connecting vertices in $\mathcal{X}$, to vertices in $\cup_{i=1}^7 K_i$ is therefore at least $\sum_{j=1}^{2n^2} \min\{|\mathbf{y}_j|, 1 +\lceil\left(|\mathbf{y}_j|-1\right)/2\rceil\}$, under the constraint that $\sum_{j=1}^{2n^2} \mathbf{y}_j[i]= 4\sqrt{M\delta^i}$.
		Without loss of generality, let us assume that $\delta_1 \geq \delta_2\geq \ldots\geq \delta_7$. As previously stated, it is possible to connect all vertices in $K_1$ to vertices in $\mathcal{X}$ through vertex disjoint paths. Consider now all possible dispositions of the vertices in $\cup_{i=2}^7 K_i$ over the outputs of the $2n^2$ encoder sub-CDAGs. 
		Recall that the output vertices of an encoder sub-CDAG belong each to a different $H^{n\times n}$ sub-CDAG. From Lemma~\ref{lem:conneconder}, we have that for each encoder, there exists a subset $X_j\subset{X}$ of the input vertices of the $j$-th encoder sub-CDAG, with
		$
			|X_j|\geq \min \Big\{|\mathbf{y}_j|, 1 + \left\lceil\left(|\mathbf{y}_j|-1\right)/2\right\rceil \Big\} \geq \mathbf{y}_j[1] + \left(\sum_{i=2}^7 \mathbf{y}_j[i]\right)/2
		$, 		for which is possible to connect all vertices in $X_j$ to $|X_j|$ \emph{distinct} output vertices of the $j$-th encoder sub-CDAG which are in $\cup_{i=1}^7 K_i$ using $|X_j|$, thus using vertex disjoint paths.
		As all the $Enc$ sub-CDAGs are vertex disjoint, we can sum their contributions and we can therefore conclude that the number of vertex disjoint paths connecting values in $\mathcal{X}$ to vertices in  $\cup_{i=1}^7 K_i$ is at least $|K_1|+ \frac{1}{2}\sum\limits_{i=2}^7|K_i| = 4\sqrt{M}\left(\sqrt{\delta_1} + \frac{1}{2}\sum\limits_{i=2}^7\sqrt{\delta_i} \right)$.
	Squaring this quantity leads to:
	\begin{equation*}
		\left(4\sqrt{M}\left(\sqrt{\delta_1} +\frac{1}{2}\sum\limits_{i=2}^7\sqrt{\delta_i} \right)\right)^2 = 16M\left(\delta_1 +\sqrt{\delta_1}\sum\limits_{i=2}^7\sqrt{\delta_i} + \left(\frac{1}{2}\sum\limits_{i=2}^7\sqrt{\delta_i}\right)^2\right).
	\end{equation*}
	As, by assumption, $\delta_{1} \geq \delta_2\geq\ldots \delta_7$, we have: $\sqrt{\delta_1}\sqrt{\delta_i}\geq \delta_i$ for $i=2,3,\ldots,7$. Thus:
	\begin{equation*}
		\left(4\sqrt{M}\left(\sqrt{\delta_1} +\frac{1}{2}\sum\limits_{i=2}^7\sqrt{\delta_i} \right)\right)^2 
		\geq 16 M \sum\limits_{i=1}^7 \delta_i \geq \left(4\sqrt{M\left(|Z|-2|\Gamma|\right)}\right)^2.
	\end{equation*}
	
	There are therefore at least $4\sqrt{M\left(|Z|-2|\Gamma|\right)}$ vertex disjoint paths connecting vertices in $\mathcal{X}$ to vertices in $\cup_{i=2}^7 K_i$. The lemma follows.
\end{proof}

\begin{lemma}\label{lem:stra_part2}
Let $H^{n\times n}$ be the CDAG which corresponds to the execution of Strassen's matrix multiplication algorithm for input matrices $A,B \in \ri^{n^2}$, with $n\geq 2\sqrt{M}$. Let $\mathcal{Y}$ (resp., $\mathcal{Z}$) denote the set of input (resp., output) vertices of the $\left(n/2\sqrt{M}\right)^{\log_2 7}$ sub-CDAGs $H^{2\sqrt{M}\times 2\sqrt{M}}$. Any dominator set for any subset $Z\subseteq \mathcal{Z}$, has size at least $|Z|/2 = 2M$.
\end{lemma}
\begin{proof}
	The proof is by contradiction. Let $\Gamma$ be a dominator set for $Z$ in $H^{n\times n}$ such that $|\Gamma|\leq 2M-1$. 
	Let $\Gamma'\subseteq \Gamma$ be the subset of ``internal'' vertices of any of any of the sub-CDAGs $H^{2\sqrt{M}\times 2\sqrt{M}}$ in $\Gamma$.	Let $\mathcal{X}$ denote the set the input vertices of $H^{n\times n}$. From Lemma~\ref{lem:stra_part1} we have that there exist at least $4\sqrt{M\left(|Z|-2|\Gamma'|\right)}$  paths connecting a subset $X \subseteq \mathcal{X}$ of global input vertices to vertices in $Z$ which do not include any vertex in $\Gamma'$. Further, the sub-paths of such paths that connect vertices in $X$ to vertices in $Y\subseteq \mathcal{Y}$ are vertex disjoint among themselves. This implies that each of the vertices in $\Gamma\setminus\Gamma'$  can therefore be traversed by at most one of these paths. The number of directed paths from $X$ to $Z$ which do not traverse any vertex in $\Gamma$ will therefore be at least $4\sqrt{M\left(|Z|-2|\Gamma'|\right)} - \left(|\Gamma| - |\Gamma'|\right)$. We have:
	\begin{align*}
		& \left(4\sqrt{M\left(|Z|-2|\Gamma'|\right)} - \left(|\Gamma| - |\Gamma'|\right)\right)^2 \\
		&\ \ \ \ \ \ \ \ \ \ \ = 16 M \left(|Z|-2|\Gamma'|\right) + \left(|\Gamma| - |\Gamma'|\right)^2 - 8\sqrt{M\left(|Z|-2|\Gamma'|\right)}\left(|\Gamma| - |\Gamma'|\right) \\
		&\ \ \ \ \ \ \ \ \ \ \ \geq 16 M \left(|Z|-2|\Gamma'|\right) - 8\sqrt{M\left(|Z|-2|\Gamma'|\right)\left(|\Gamma| - |\Gamma'|\right)} \\
		&\ \ \ \ \ \ \ \ \ \ \ \geq 16 M \left(|Z|-2|\Gamma|\right) + 32 M \left(|\Gamma| - |\Gamma'|\right) - 8\sqrt{M\left(|Z|-2|\Gamma'|\right)}\left(|\Gamma| - |\Gamma'|\right) \\
		&\ \ \ \ \ \ \ \ \ \ \ \geq 16 M \left(|Z|-2|\Gamma|\right) +\left(|\Gamma| - |\Gamma'|\right) \left(32M - 8\sqrt{M\left(|Z|-2|\Gamma'|\right)}\right)
	\end{align*}
Since by hypothesis $|Z|-2|\Gamma'|\leq 4M$, we have $	\left(4\sqrt{M\left(|Z|-2|\Gamma'|\right)} - \left(|\Gamma| - |\Gamma'|\right)\right)^2 \geq 16 M \left(|Z|-2|\Gamma|\right)$,
and thus $ 4\sqrt{M\left(|Z|-2|\Gamma'|\right)} - \left(|\Gamma| - |\Gamma'|\right) \geq 4\sqrt{M\left(|Z|-2|\Gamma|\right)}$.

For $|\Gamma|\leq 2M-1$ there is therefore at least one path connecting an input vertex to a vertex in $Z$ which does not traverse any vertex in $\Gamma$. This constitutes a contradiction.
\end{proof}

Lemma~\ref{lem:stra_part2} provides us with the tools required to obtain our main result. 

\begin{theorem}[Lower bound \io complexity Strassen's algorithm]
\label{thm:genstrass}
Consider Strassen's algorithm being used to multiply two square matrices $A,B \in\ri^{n\times n}$. 
The I/O-complexity of Strassen's algorithm when run on a sequential machine with a cache of size $M$ is:
\begin{equation}\label{eq:stramain}
	IO_{H^{n\times n}}\left(M\right) \geq \frac{1}{7}\left(\frac{n}{\sqrt{M}}\right)^{\log_2 7}M.
\end{equation}
If run on $P$ processors each equipped with a local memory of size $M$ the \io complexity is:
\begin{equation}\label{eq:strapar}
	IO_{H^{n\times n}}(P,M) \geq \frac{1}{7} \left( \frac{n}{\sqrt{M}}\right)^{\log_2 7} \frac{M}{P} 
\end{equation}
\end{theorem}
\begin{proof}
We start by proving (\ref{eq:seqnewnr}). We assume without loss of generality that $n = 2^a$ and $\sqrt{M}=2^b$  for some $a,b\in \mathbb{N}$. At least $3M$ \io operations are necessary in order to read all the $2n^2$ input values form slow memory to the cache and to write the $n^2$ output values to the slow memory once they have been computed. The theorem is therefore verified if $n\leq 2\sqrt{M}$. 

For $n\geq 4\sqrt{M}$, let $\mathcal{Z}$ denote the set of output vertices of the $\left(n/2\sqrt{M}\right)^{\log_2 7}$ sub-CDAGs $H^{2\sqrt{M}\times 2\sqrt{M}}$ of $H^{n\times n}$.
	Let $\mathcal{C}$ be any computation schedule for the sequential execution of Strassen's algorithm using a cache memory of size $M$. 
	We partition $\mathcal{C}$ into segments such that during each segment $\mathcal{C}_i$ the values corresponding to exactly $4M$ distinct vertices in $\mathcal{Z}$ (denoted as $Z_i$) are computed for the \emph{first time}. As we have $|\mathcal{Z}|= 4M \left(\frac{n}{2\sqrt{M}}\right)^{\log 7}$, there will be $\left(n/\left(2\sqrt{M}\right)\right)^{\log 7}$ such segments.
%
	We shall now verify that at least $M$ \io operations are to be executed during every segment $\mathcal{C}_i$. The proof is by contradiction.   
	 Let $\Gamma_i$ denote the set of vertices of $H^{n \times n}$ which correspond to the at most $M$ values which are stored in the cache at the beginning of the $i$-th segment and to the at most $M-1$ values loaded into the cache form the slow memory during $\mathcal{C}_i$ by means of a \emph{read} \io operation. We then have $|\Gamma_i|\leq 2M -1$.
	 In order for the $4M$ values from $Z_i$ to be computed during the segment without any additional \io operation, there must be no path connecting any vertex in $Z_i$ to any input vertex of $H^{n\times n}$ which does not have at least one vertex in $\Gamma_i$ (i.e. $\Gamma_i$ has to be a \emph{dominator set} of $Z_i$). From Lemma~\ref{lem:stra_part2}, we have that any sub-set of $4M$ elements of $\mathcal{Z}$ has dominator size at least $2M$. This leads to a contradiction.
	
	At least $M$ \io operations are thus executed during each segment $\mathcal{C}_i$. Since, by construction, the $\left(\frac{n}{2\sqrt{M}}\right)^{\log 7}$ segments are not overlapping, we can therefore conclude that at least $M\left(\frac{n}{2\sqrt{M}}\right)^{\log 7}$ \io operations are necessary for the execution of any computation schedule $\mathcal{C}$. This concludes the proof for lower bound for the sequential case in (\ref{eq:stramain}). \\
	The proof for the bound for the parallel model in equation (\ref{eq:strapar}), follows a similar strategy: at least one of the $P$ processors being used, denoted as $P^*$, must compute at least $|\mathcal{Z}|/P = 4M \left(\frac{n}{2\sqrt{M}}\right)^{\log 7}/P$ values corresponding to vertices in $\mathcal{Z}$. The bound follows by applying the same argument discussed for the sequential case to the computation executed by $P^*$. 
	\end{proof}

In~\cite{ballard2012communicationalg}, Ballard et al. presented a version of Strassen's algorithm whose \io cost matches, up to a constant the lower bound obtained in Theorem~\ref{thm:genstrass}. This implies that our bound is indeed asymptotically tight and that  the version of Strassen's algorithm presented by Ballard et al. in~\cite{ballard2012communicationalg}, is indeed asymptotically optimal with respect to the \io cost. Furthermore, as in the optimal algorithm presented in~\cite{ballard2012communicationalg} no intermediate result is ever recomputed, we can conclude the use of recomputation can lead at most to a constant factor reduction of the \io complexity for the execution of Strassen's matrix multiplication algorithm. 
	
\section{Conclusions}\label{sec:conlcusion}

This work has contributed to the characterization of the \io
complexity of Strassen's algorithm, by establishing asymptotically
tight lower bounds that hold even when recomputation is allowed.
The technique we have developed crucially exploits the recursive
nature of the CDAG, which makes it promising for the analysis of other
recursive algorithms, beginning with fast rectangular matrix
multiplication algorithms~\cite{le2012faster}.

The relationship we have exploited between dominator size and
Grigoriev's flow points at connections between \io complexity,
(pebbling) space-time tradeoffs~\cite{savage97models}, and VLSI
area-time tradeoffs~\cite{thompson79vlsistoc}; these connections
deserve further attention.

Some CDAGs for which non trivial \io complexity lower bounds are known
only in the case of no recomputations are described
in~\cite{bilardi1999processor}. These CDAGs are of particular interest
as examples of speedups superlinear with the number of processors, in
the ``\emph{limiting technology}'', defined by fundamental limitations on
device size and message speed.  Whether such speedups hold even when
recomputation is allowed remains an open question, which the
techniques introduced here might help answer.

In general, while it is well known that recomputation may reduce the
\io complexity of some CDAGs, we are far from a characterization of
those CDAGs for which recomputation is effective. This broad goal
remains a challenge for any attempt toward a general theory of the
communication requirements of computations.
\subparagraph*{Acknowledgements.}

This work was supported, in part, by MIUR of Italy under project
AMANDA 2012C4E3KT 004 and by the University of Padova under projects
STPD08JA32, CPDA121378/12, and CPDA152255/15.



\bibliography{bibliography}
\clearpage

\appendix
\section{Properties of Strassen's algorithm}

\begin{algorithm}
\caption{Strassen's Matrix Multiplication}\label{alg:strass}
\begin{algorithmic}[1]
	\Statex \textbf{Input:} matrices $A,B$
	\Statex \textbf{Output:} matrix $C$
	\Procedure{StrassenMM}{A,B}
	\If{$n=1$}
	\State $C=A\cdot B$
	\Else
	\State Decompose $A$ and $B$ into four equally sized bloc matrices as follows:
	\Statex \begin{equation*}
\begin{array}{cc}
	A = \left[ \begin{array}{cc}
	A_{1,1} &A_{1,2} \\
	A_{2,1} &A_{2,2}
	\end{array}
	\right], & B = \left[ \begin{array}{cc}
	B_{1,1} &B_{1,2} \\
	B_{2,1} &B_{2,2}
	\end{array}
	\right]
\end{array}
\end{equation*}
\State $M_1 = \textsc{StrassenMM}\left(A_{1,1} + A_{2,2},B_{1,1} + B_{2,2}\right)$
\State $M_2 = \textsc{StrassenMM}\left(A_{2,1} + A_{2,2}, B_{1,1}\right)$
\State $M_3 = \textsc{StrassenMM}\left(A_{1,1},B_{1,2} - B_{2,2}\right)$
\State $M_4 = \textsc{StrassenMM}\left(A_{2,2},B_{2,1} - B_{1,1}\right)$
\State $M_5 = \textsc{StrassenMM}\left(A_{1,1} + A_{1,2}, B_{2,2}\right)$
\State $M_6 = \textsc{StrassenMM}\left(A_{2,1} - A_{1,1}, B_{1,1} + B_{1,2}\right)$
\State $M_7 = \textsc{StrassenMM}\left(A_{1,2} - A_{2,2}, B_{2,1} + B_{2,2}\right)$
\State $C_{1,1} = M_1 +M_4-M_5+M_7$ 
\State $C_{1,2} = M_3 + M_5$
\State $C_{2,1} = M_2 + M_4$
\State $C_{2,2} = M_1 -M_2 +M_3 +M_6$
\EndIf
\Return $C$
\EndProcedure
\end{algorithmic}
\end{algorithm}

The original version of  Strassen's fast matrix multiplication~\cite{strassen1969gaussian} is reported in Algortihm~\ref{alg:strass}. We refer the reader to~\cite{winograd1971multiplication} for
Winograd's variant, which reduces the number of additions. 

\label{app:proof_encoder}
\begin{proof}[Proof of Lemma~\ref{lem:disictness_sub_cdag}]
	Let us consider a single recursive step. While some of the sub-problems may use as input values form either $A$ and $B$ without any combination (i.e., $M_3$ uses the sub-matrix $A_{1,1}$ as input), none of the seven sub-problems share any of their input values. As this consideration holds for each level of the recursive construction, we have that the $7^i$ sub-problems with input of size $\frac{n}{2^i}\times\frac{n}{2^i}$ generated by of Strassen's algorithm do not share any input value.
\end{proof}

\begin{proof}[Proof of Lemma~\ref{lem:conneconder}]
We provide the proof for $Enc_A$, as $Enc_A$ and $Enc_B$ are isomorphic,  the result holds for $Enc_B$ as well. 
	Note that in $Enc_A$ there are some pairs of input-output vertices $u,v$ for which the input vertex $u$ is the only predecessor of the output vertex $v$. This implies that the two vertices are actually one unique vertex. With a little abuse of notation, we will still say that $u$  can be connected to $v$ via a single edge.
		
	We assign an index to each of the output vertices of $Enc_A$ according to how is indicated in Figure~\ref{fig:enca}.  
Note that the index assigned to each output corresponds to the index of the sub-problem generated by Strassen's algorithm for which the corresponding value is used as an input (see Figure~\ref{fig:bsestra} and Figure~\ref{fig:strarec} in Section~\ref{sec:preliminaries}). 

In order to verify that this lemma holds, we study all possible compositions of a subset $Y$ of the output vertices of $Enc_A$.  Each of these compositions is identified by a vector $\mathbf{y}$ with seven components, where $y_i = 1$ if the $i$-th output of $Enc_A$ is in $Y$ or zero otherwise, for $i\in\{1,2,\ldots,7\}$. In order to improve the presentation, we associate to each of the possible 128 compositions a \emph{code} given by $\sum_{i=1}^7 y_i2^{7-i}$. In Table~\ref{table:configurations}, we study each of the 128 possible compositions of $Y$, which are ordered by the value of $|Y|$ and by their code. The value in the last column $|X|$ denotes the maximum size of a sub-set $X$ of the input vertices of $Enc_A$ such that each vertex in $X$ can be connected to a distinct vertex in the subset $Y$ corresponding to $\mathbf{y}$. Each of these values can be obtained through a straightforward analysis of $Enc_A$.
The lemma then follows from observing that for any possible composition of the set $Y$ we have
\begin{equation*}
	|X|\geq \min \Big\{|Y|, 1 + \lceil\left(|Y|-1\right)/2\rceil\} .
\end{equation*}
\end{proof}

\begin{proof}[Proof of Lemma~\ref{lem:vdigraphslike}]
A proof for this result was presented in~\cite{scott2015matrix} (Lemma 1). For completeness, we present here a version of the proof using our notation.

Given an $\left(n_0, m_0\right)$-Strassen like algorithm used to multiply $A,B\in\ri^{n^2}$, let $\stral{n\times n}$ denote the corresponding CDAG. The recursive structure of the algorithm will generate a total of $m_0^{\log_{n_0}n-i} $ sub-problems with input size $n/n_0^i$ each of which corresponds to a sub-CDAGs $\stral{n/n_0^i\times n/n_0^i}$.

The lemma is trivially verified for $0\leq i \leq 1$, in the following we consider the case for $2\leq i\leq \log_{n_0} m_0$. As by hypothesis, non-trivial combinations of the input matrices are used as input for at most one sub-problem, the only input values which two sub-problems of the same size  can  share are sub-matrices of the ``\emph{global}'' input matrices $A,B$. For $i'=i-2$, et $\mathcal{P}_{1}$ be one of the sub-computations corresponding to one of the sub-CDAGs $\stral{n/n_0^{i'}\times n/n_0^{i'}}$ which multiplies matrices $A_1$ and $B_1$. 
Without loss of generality we assume that for both the basic encoder CDAGs for $A$ and $B$ at least one of the outputs is given by a non-trivial linear combination of elements of the input matrix. It is in fact well known that any algorithm which computes linear combinations of only one of the input matrices performs no better than the na\'ive-definition based matrix multiplication and it is therefore not a Strassen-like fast matrix multiplication algorithm.
This implies that at least one of the sub-problems $\mathcal{P}_{2}$ generated by $\mathcal{P}_{1}$ multiplies matrices $A_2$ and $B_2$ such that $A_2$ is a nontrivial linear combination of sub-matrices of $A_1$. Similarly, at least one of the  sub-computations $\mathcal{P}_{3}$ generated by $\mathcal{P}_{2}$ multiplies matrices $A_3$ by $B_3$ such that $B_3$ is  non-trivial linear combination of sub-matrices of $B$. $A_3$ (resp., $B_3$) shares no input with $A$ (resp., B). Thus at least one of the $m_0^2$ sub-computations for input size $n/n_0^i \times n/n_0^i$  generated by $\mathcal{P}_{1}$ is input-disjoint from it. 
As by hypothesis, non-trivial combinations of the input matrices are used as input for at most one sub-problem, all the sub-computations $\mathcal{P}_{3}$ corresponding to each $\mathcal{P}_{1}$ are input-disjoint and the corresponding sub-CDAGs $\stral{n/n_0^i\times n/n_0^i}$ are thus vertex disjoint.
\end{proof}
	\begin{longtable}[c]{|p{.04\textwidth}||p{.02\textwidth}||p{.03\textwidth} p{.03\textwidth} p{.03\textwidth} p{.03\textwidth} p{.03\textwidth} p{.03\textwidth} p{.03\textwidth}|p{.02\textwidth}|}
		\caption{Study of the possible compositions of sub-sets of output vertices of $Enc$ for Lemma~\ref{lem:conneconder}}\label{table:configurations}\\ 
		\hline
		code & $|\textbf{y}|$ & $y_1$ & $y_2$ & $y_3$ & $y_4$ & $y_5$ & $y_6$ & $y_7$ &  $|X|$ \\\hline
		\endfirsthead
		\hline
		\multicolumn{10}{|c|}{{Table~\ref{table:configurations} - continued from previous page}}\\\hline
		code & $|\textbf{y}|$ & $y_1$ & $y_2$ & $y_3$ & $y_4$ & $y_5$ & $y_6$ & $y_7$ & $|X|$ \\\hline
		\endhead
		
		\hline
		\multicolumn{10}{|c|}{{Continued on next page}}\\\hline
		\endfoot	
		
\hline 
\endlastfoot
			0   & 0 & 0 & 0 & 0 & 0 & 0 & 0 & 0 & 0 \\\hline
			1   & 1 & 0 & 0 & 0 & 0 & 0 & 0 & 1 & 1 \\
			2   & 1 & 0 & 0 & 0 & 0 & 0 & 1 & 0 & 1 \\
			4   & 1 & 0 & 0 & 0 & 0 & 1 & 0 & 0 & 1 \\
			8   & 1 & 0 & 0 & 0 & 1 & 0 & 0 & 0 & 1 \\
			16  & 1 & 0 & 0 & 1 & 0 & 0 & 0 & 0 & 1 \\
			32  & 1 & 0 & 1 & 0 & 0 & 0 & 0 & 0 & 1 \\
			64  & 1 & 1 & 0 & 0 & 0 & 0 & 0 & 0 & 1 \\\hline
			3   & 2 & 0 & 0 & 0 & 0 & 0 & 1 & 1 & 2 \\
			5   & 2 & 0 & 0 & 0 & 0 & 1 & 0 & 1 & 2 \\
			6   & 2 & 0 & 0 & 0 & 0 & 1 & 1 & 0 & 2 \\
			9   & 2 & 0 & 0 & 0 & 1 & 0 & 0 & 1 & 2 \\
			10  & 2 & 0 & 0 & 0 & 1 & 0 & 1 & 0 & 2 \\
			12  & 2 & 0 & 0 & 0 & 1 & 1 & 0 & 0 & 2 \\
			17  & 2 & 0 & 0 & 1 & 0 & 0 & 0 & 1 & 2 \\
			18 	& 2 & 0 & 0 & 1 & 0 & 0 & 1 & 0 & 2 \\
			20  & 2 & 0 & 0 & 1 & 0 & 1 & 0 & 0 & 2 \\
			24  & 2 & 0 & 0 & 1 & 1 & 0 & 0 & 0 & 2 \\
			33  & 2 & 0 & 1 & 0 & 0 & 0 & 0 & 1 & 2 \\
			34  & 2 & 0 & 1 & 0 & 0 & 0 & 1 & 0 & 2 \\
			36  & 2 & 0 & 1 & 0 & 0 & 1 & 0 & 0 & 2 \\
			40  & 2 & 0 & 1 & 0 & 1 & 0 & 0 & 0 & 2 \\
			48  & 2 & 0 & 1 & 1 & 0 & 0 & 0 & 0 & 2 \\
			65  & 2 & 1 & 0 & 0 & 0 & 0 & 0 & 1 & 2 \\
			66  & 2 & 1 & 0 & 0 & 0 & 0 & 1 & 0 & 2 \\
			68  & 2 & 1 & 0 & 0 & 0 & 1 & 0 & 0 & 2 \\
			72  & 2 & 1 & 0 & 0 & 1 & 0 & 0 & 0 & 2 \\
			80  & 2 & 1 & 0 & 1 & 0 & 0 & 0 & 0 & 2 \\
			96  & 2 & 1 & 1 & 0 & 0 & 0 & 0 & 0 & 2 \\\hline
			7   & 3 & 0 & 0 & 0 & 0 & 1 & 1 & 1 & 3 \\
			11  & 3 & 0 & 0 & 0 & 1 & 0 & 1 & 1 & 3 \\
			13  & 3 & 0 & 0 & 0 & 1 & 1 & 0 & 1 & 3 \\
			14  & 3 & 0 & 0 & 0 & 1 & 1 & 1 & 0 & 3 \\		
			19  & 3 & 0 & 0 & 1 & 0 & 0 & 1 & 1 & 3 \\
			21  & 3 & 0 & 0 & 1 & 0 & 1 & 0 & 1 & 3 \\
			22  & 3 & 0 & 0 & 1 & 0 & 1 & 1 & 0 & 3 \\
			25  & 3 & 0 & 0 & 1 & 1 & 0 & 0 & 1 & 3 \\
			26  & 3 & 0 & 0 & 1 & 1 & 0 & 1 & 0 & 3 \\
			28  & 3 & 0 & 0 & 1 & 1 & 1 & 0 & 0 & 3 \\
			35  & 3 & 0 & 1 & 0 & 0 & 0 & 1 & 1 & 3 \\
			37  & 3 & 0 & 1 & 0 & 0 & 1 & 0 & 1 & 3 \\
			38  & 3 & 0 & 1 & 0 & 0 & 1 & 1 & 0 & 3 \\
			41  & 3 & 0 & 1 & 0 & 1 & 0 & 0 & 1 & 3 \\
			42  & 3 & 0 & 1 & 0 & 1 & 0 & 1 & 0 & 3 \\
			44  & 3 & 0 & 1 & 0 & 1 & 1 & 0 & 0 & 3 \\
			49  & 3 & 0 & 1 & 1 & 0 & 0 & 0 & 1 & 3 \\
			50  & 3 & 0 & 1 & 1 & 0 & 0 & 1 & 0 & 3 \\
			52  & 3 & 0 & 1 & 1 & 0 & 1 & 0 & 0 & 3 \\
			56  & 3 & 0 & 1 & 1 & 1 & 0 & 0 & 0 & 3 \\
			67  & 3 & 1 & 0 & 0 & 0 & 0 & 1 & 1 & 3 \\
			69  & 3 & 1 & 0 & 0 & 0 & 1 & 0 & 1 & 3 \\
			70  & 3 & 1 & 0 & 0 & 0 & 1 & 1 & 0 & 3 \\
			73  & 3 & 1 & 0 & 0 & 1 & 0 & 0 & 1 & 3 \\
			74  & 3 & 1 & 0 & 0 & 1 & 0 & 1 & 0 & 3 \\		
			76  & 3 & 1 & 0 & 0 & 1 & 1 & 0 & 0 & 3 \\
			81  & 3 & 1 & 0 & 1 & 0 & 0 & 0 & 1 & 3 \\
			82  & 3 & 1 & 0 & 1 & 0 & 0 & 1 & 0 & 3 \\
			84 & 3 & 1 & 0 & 1 & 0 & 1 & 0 & 0 & 3 \\
			88 & 3 & 1 & 0 & 1 & 1 & 0 & 0 & 0 & 2 \\
			97 & 3 & 1 & 1 & 0 & 0 & 0 & 0 & 1 & 3 \\
			98 & 3 & 1 & 1 & 0 & 0 & 0 & 1 & 0 & 3 \\
			100 & 3 & 1 & 1 & 0 & 0 & 1 & 0 & 0 & 3 \\
			104 & 3 & 1 & 1 & 0 & 1 & 0 & 0 & 0 & 3 \\
			112 & 3 & 1 & 1 & 1 & 0 & 0 & 0 & 0 & 3 \\\hline
			15  & 4 & 0 & 0 & 0 & 1 & 1 & 1 & 1 & 4 \\
			23  & 4 & 0 & 0 & 1 & 0 & 1 & 1 & 1 & 4 \\
			27  & 4 & 0 & 0 & 1 & 1 & 0 & 1 & 1 & 4 \\
			29 	& 4 & 0 & 0 & 1 & 1 & 1 & 0 & 1 & 3 \\
			30  & 4 & 0 & 0 & 1 & 1 & 1 & 1 & 0 & 4 \\
			39  & 4 & 0 & 1 & 0 & 0 & 1 & 1 & 1 & 4 \\
			43  & 4 & 0 & 1 & 0 & 1 & 0 & 1 & 1 & 4 \\
			45  & 4 & 0 & 1 & 0 & 1 & 1 & 0 & 1 & 4 \\
			46  & 4 & 0 & 1 & 0 & 1 & 1 & 1 & 0 & 4 \\
			51  & 4 & 0 & 1 & 1 & 0 & 0 & 1 & 1 & 4 \\
			53  & 4 & 0 & 1 & 1 & 0 & 1 & 0 & 1 & 4 \\
			54  & 4 & 0 & 1 & 1 & 0 & 1 & 1 & 0 & 4 \\
			57  & 4 & 0 & 1 & 1 & 1 & 0 & 0 & 1 & 4 \\
			58  & 4 & 0 & 1 & 1 & 1 & 0 & 1 & 0 & 3 \\
			60	& 4 & 0 & 1 & 1 & 1 & 1 & 0 & 0 & 4 \\
			71  & 4 & 1 & 0 & 0 & 0 & 1 & 1 & 1 & 4 \\
			75  & 4 & 1 & 0 & 0 & 1 & 0 & 1 & 0 & 3 \\
			77  & 4 & 1 & 0 & 0 & 1 & 1 & 0 & 1 & 3 \\
			78  & 4 & 1 & 0 & 0 & 1 & 1 & 1 & 0 & 4 \\
			83  & 4 & 1 & 0 & 1 & 0 & 0 & 1 & 1 & 4 \\
			85  & 4 & 1 & 0 & 1 & 0 & 1 & 0 & 1 & 3 \\
			86  & 4 & 1 & 0 & 1 & 0 & 1 & 1 & 0 & 4 \\
			89  & 4 & 1 & 0 & 1 & 1 & 0 & 0 & 1 & 3 \\
			90  & 4 & 1 & 0 & 1 & 1 & 0 & 1 & 0 & 3 \\
			92  & 4 & 1 & 0 & 1 & 1 & 1 & 0 & 0 & 3 \\
			99  & 4 & 1 & 1 & 0 & 0 & 0 & 1 & 1 & 4 \\
			101 & 4 & 1 & 1 & 0 & 0 & 1 & 0 & 1 & 4 \\
			102 & 4 & 1 & 1 & 0 & 0 & 1 & 1 & 0 & 4 \\	
			105 & 4 & 1 & 1 & 0 & 1 & 0 & 0 & 1 & 4 \\
			106 & 4 & 1 & 1 & 0 & 1 & 0 & 1 & 0 & 3 \\
			108 & 4 & 1 & 1 & 0 & 1 & 1 & 0 & 0 & 4 \\
			113 & 4 & 1 & 1 & 1 & 0 & 0 & 0 & 1 & 4 \\
			114 & 4 & 1 & 1 & 1 & 0 & 0 & 1 & 0 & 3 \\
			116 & 4 & 1 & 1 & 1 & 0 & 1 & 0 & 0 & 4 \\
			120 & 4 & 1 & 1 & 1 & 1 & 0 & 0 & 0 & 3 \\\hline 
			31  & 5 & 0 & 0 & 1 & 1 & 1 & 1 & 1 & 4 \\
			47  & 5 & 0 & 1 & 0 & 1 & 1 & 1 & 1 & 4 \\
			55  & 5 & 0 & 1 & 1 & 0 & 1 & 1 & 1 & 4 \\
			59  & 5 & 0 & 1 & 1 & 1 & 0 & 1 & 1 & 4 \\
			61  & 5 & 0 & 1 & 1 & 1 & 1 & 0 & 1 & 4 \\
			62  & 5 & 0 & 1 & 1 & 1 & 1 & 1 & 0 & 4 \\
			79  & 5 & 1 & 0 & 0 & 1 & 1 & 1 & 1 & 4 \\
			87  & 5 & 1 & 0 & 1 & 0 & 1 & 1 & 1 & 4 \\
			91  & 5 & 1 & 0 & 1 & 1 & 0 & 1 & 1 & 4 \\
			93  & 5 & 1 & 0 & 1 & 1 & 1 & 0 & 1 & 3 \\
			94  & 5 & 1 & 0 & 1 & 1 & 1 & 1 & 0 & 4 \\
			103 & 5 & 1 & 1 & 0 & 0 & 1 & 1 & 1 & 4 \\
			107 & 5 & 1 & 1 & 0 & 1 & 0 & 1 & 1 & 4 \\
			109 & 5 & 1 & 1 & 0 & 1 & 1 & 0 & 1 & 4 \\
			110 & 5 & 1 & 1 & 0 & 1 & 1 & 1 & 0 & 4 \\
			115 & 5 & 1 & 1 & 1 & 0 & 0 & 1 & 1 & 4 \\
			117 & 5 & 1 & 1 & 1 & 0 & 1 & 0 & 1 & 4 \\
			118 & 5 & 1 & 1 & 1 & 0 & 1 & 1 & 0 & 4 \\
			121 & 5 & 1 & 1 & 1 & 1 & 0 & 0 & 1 & 4 \\
			122 & 5 & 1 & 1 & 1 & 1 & 0 & 1 & 0 & 3 \\
			124 & 5 & 1 & 1 & 1 & 1 & 1 & 0 & 0 & 4 \\\hline
			63  & 6 & 0 & 1 & 1 & 1 & 1 & 1 & 1 & 4 \\
			95  & 6 & 1 & 0 & 1 & 1 & 1 & 1 & 1 & 4 \\
			111 & 6 & 1 & 1 & 0 & 1 & 1 & 1 & 1 & 4 \\
			119 & 6 & 1 & 1 & 1 & 0 & 1 & 1 & 1 & 4 \\
			123 & 6 & 1 & 1 & 1 & 1 & 0 & 1 & 1 & 4 \\
			125 & 6 & 1 & 1 & 1 & 1 & 1 & 0 & 1 & 4 \\
			126 & 6 & 1 & 1 & 1 & 1 & 1 & 1 & 0 & 4 \\\hline
			127 & 7 & 1 & 1 & 1 & 1 & 1 & 1 & 1 & 4 	\\\hline
\end{longtable}

\end{document}